\documentclass[letter]{article}[11pt]

\usepackage{amssymb, amsmath, setspace, epsfig,graphicx,psfrag,color}

\definecolor{Red}{rgb}{1,0,0}
\definecolor{Blue}{rgb}{0,0,1}
\definecolor{Olive}{rgb}{0.41,0.55,0.13}
\definecolor{Green}{rgb}{0,1,0}
\definecolor{MGreen}{rgb}{0,0.8,0}
\definecolor{DGreen}{rgb}{0,0.55,0}
\definecolor{Yellow}{rgb}{1,1,0}
\definecolor{Cyan}{rgb}{0,1,1}
\definecolor{Magenta}{rgb}{1,0,1}
\definecolor{Orange}{rgb}{1,.5,0}
\definecolor{Violet}{rgb}{.5,0,.5}
\definecolor{Purple}{rgb}{.75,0,.25}
\definecolor{Brown}{rgb}{.75,.5,.25}
\definecolor{Grey}{rgb}{.5,.5,.5}
\definecolor{Pink}{rgb}{1,0,1}
\definecolor{DBrown}{rgb}{.5,.34,.16}

\definecolor{Black}{rgb}{0,0,0}

\newtheorem{lemma}{Lemma}
\newtheorem{theorem}{Theorem}

\newenvironment{proof}[1][Proof]{\begin{trivlist}
\item[\hskip \labelsep {\bfseries #1}]}{\end{trivlist}}
\newenvironment{definition}[1][Definition]{\begin{trivlist}
\item[\hskip \labelsep {\bfseries #1}]}{\end{trivlist}}

\newcommand{\qed}{\nobreak \ifvmode \relax \else
      \ifdim\lastskip<1.5em \hskip-\lastskip
      \hskip1.5em plus0em minus0.5em \fi \nobreak
      \vrule height0.75em width0.5em depth0.25em\fi}

\newcommand{\enp} {\hfill \rule{2.2mm}{2.6mm}}

\newcommand{\la}{\lambda}

\newcommand{\ep}{\epsilon}

\newcommand{\mc}{\mathcal}

\newcommand{\mrm}{\mathrm}

\newcommand{\M}{\mathrm{M}}
\newcommand{\PM}{\mathrm{PM}}

\title{On the exactness of the cavity method for Weighted $\mrm{b}$-Matchings on Arbitrary Graphs and its Relation to Linear Programs}

\author{Mohsen Bayati\thanks{Microsoft Research;
\{mohsenb~,~borgs~,~jchayes\}@microsoft.com}\and
Christian Borgs$^*$ \and Jennifer Chayes$^*$ \and Riccardo
Zecchina\thanks{Politecnico Di Torino; riccardo.zecchina@polito.it}}

\date{}
\begin{document}

 \maketitle

\begin{abstract}
We consider the general problem of finding
the minimum weight $\mrm{b}$-matching on arbitrary graphs.
We prove that, whenever the linear
programming relaxation of the problem has no fractional solutions,
then the cavity  or belief propagation equations converge to the
correct solution both for synchronous and asynchronous updating.

\end{abstract}

Motivated by the cavity method , very fast distributed heuristic algorithms have recently been developed for the solution
of random constraint satisfaction problems.
In some cases, namely in the replica symmetric  scheme, the algorithms generated by the cavity method are
exactly of the form of a classic belief propagation {(max-product or min-sum)},
i.e. a message-passing algorithm for efficiently
computing marginal probabilities or finding the assignment with highest probability of a joint discrete probability
distribution defined on a graph. The belief propagation (BP) algorithm converges to
a correct solution if the associated graph is a tree, and
may be also {a good heuristic} for some graphs with cycles.
In other cases, e.g.  {when the space of solutions clusters into many subsets}, the cavity method may lead to a more involved survey propagation (SP) algorithm \cite{MeZ02,BMZ05}, in which some form of long range correlation among variables is included in the formalism.

In this paper, we study the problem of finding the minimum weight $\mrm{b}$-matchings in \emph{arbitrary} graphs
via the min-sum version of the cavity/ BP algorithm.

Let $G=(V,E)$ be an undirected graph with edge weights $w_{ij}$ for
each edge $\{i,j\}\in E$ and node capacities $b_i$ for each node
$i\in V$. The iterative message-passing algorithm based on synchronous BP for solving the weighted perfect
$\mrm{b}$-matching problem
is the following \emph{simple} procedure:
At each time, every vertex of the graph sends (real {valued}) messages  to each of its neighbors.
The message transmitted at time $t$ from $i$ to $j$ is $w_{ij}$ minus the $b_i^{th}$ minimum of
the messages previously received by
$i$ at time $t-1$ from all of its neighbors except $j$. At the end of each
iteration, every vertex $i$ selects $b_i$ of its adjacent edges that
correspond to the $b_i$ smallest received messages.
This procedure can be derived as the zero temperature limit of the cavity equations for the minimum b-matchning problem. In what follows we will use as synonymous cavity equation, belief propagation (BP) or min sum equations.

We will show the following result:  For {arbitrary} graphs $G$,
and all sets of weights {$\{w_{ij}\}$}, after $O(n)$ iterations, the set of selected edges converges to the correct solution,
i.e., to the minimum weight perfect $\mrm{b}$-matching of $G$, provided
that the Linear Programming (LP) relaxation of the problem  has no fractional solutions.
Additionally we introduce a new construction, a \emph{generalized computation tree},
which allows us to analyze the more complicated case of BP with an asynchronous updating scheme, and prove
convergence and correctness of it when each edge of the graph transmits at least $\theta(n)$ messages.

Our proof gives new insight of the often-noted but poorly understood connection between the cavity method (or BP)  and LP through the dual of the LP relaxation.

We also modify our BP algorithm and its analysis to include the problem of finding the \emph{non-perfect}
weighted $\mrm{b}$-matchings. Independent work on this aspect  can be found in \cite{SMW07}.  {
Previous exact results on \emph{bipartite graphs} was obtained in \cite{BSS05} for $\mrm{b}=1$ and then was extended to all $\mrm{b}$ in \cite{HuJ07}.}

The weighted $\mrm{b}$-matching problem is an important
problem in combinatorial optimization.  For extensive surveys see \cite{Ger95} and \cite{Pul95}.
In physics, {the} study of the random $1$-matching problem by the cavity method goes back to the work
of M\`{e}zard and Parisi \cite{MeP86} who made a celebrated conjecture for the expected
optimum weight ($\pi^2/6$) that was proven to be exact a decade later by Aldous \cite{Ald01}.
 {BP algorithms have been the subject of extensive study in several communities (see \cite{BSS05} and \cite{BBCZ07-arXiv} for a detailed survey of rigorous results about BP).}
Recent works have also suggested a connection between the BP algorithm
and linear programming (LP) in particular problems. A relationship
between iterative decoding of channel codes and LP decoding was
studied in \cite{FWK05}, \cite{VoK04}, \cite{VoK05}. Other relationships were noted in the
context of BP algorithms with convex free energies
\cite{WJW05},  \cite{WYM07}, and in the case of BP algorithms
for resource allocations \cite{MoV07J}. For weighted $1$-matchings, the connection was studied \cite{BSS06}
in the context of similarities between BP equations and the primal-dual auction algorithm of Bertsekas \cite{Ber88}.
Finally, we note that the BP equations for solving the weighted matching problem
which we use in this paper have been previously studied in \cite{BSS06}, \cite{HuJ07}.
These equations are also similar (though not identical) to the cavity equations for weighted
matching problems and traveling salesman problems which are found in the statistical physics literature,
see e.g.   {\cite{MeP86}, \cite{Ald01}, \cite{GNS05}, and } \cite{ZdM06}.

 The main contributions of our results can be summarized as follows:
\begin{enumerate}
\item The properties of the BP equations for the weighted matching was first used in \cite{BSS05}
and its correctness and convergence was shown for bipartite graphs with
unique optimum solution. The same technique was used in \cite{HuJ07} to extend the result to
$\mrm{b}$-matchings in bipartite graphs. But this technique fails for
graphs containing cycles with odd length. In order to bypass this difficulty
we use a completely different tool, complementary slackness conditions of
the LP relaxation and its dual, which is independent of the graph structure.

\item The connection between LP and the cavity method (or BP) is particularly important given that LP is a widely used technique for optimization.
We show both convergence and correctness of the BP algorithm when LP relaxation has no
fractional solutions, thus establishing a solid link between the two methods.
 {In some} recent  work, \cite{WYM07}, which generalizes methods of \cite{WJW05},
the connection of the BP algorithm and LP relaxation is studied
{in} the converged case of the BP. The authors also study interesting variations
of the BP which have convex free energies.

\item \emph{The asynchronous BP}, which includes the synchronous
version as a special case, has been a more popular version for practical purposes.
But, due to its more complicated structure,
it has not been the subject of much rigorous study. Here we provide  correctness and convergence proof of asynchronous
BP for a combinatorial optimization problem based on the construction of a generalized computation tree, which can be used for the analysis of the both convergence and correctness of  asynchronous
message-passing algorithms as the cavity method.
\end{enumerate}

Consider an undirected simple graph $G=(V,E)$, with vertices $V =
\{1,\ldots,n\}$, and edges $E$. Let each edge $\{i,j\}$ have
weight $w_{ij}\in \mathbb{R}$. Denote  {the} set of neighbors of each
vertex $i$ in $G$ by $N(i)$. Let $\mrm{b}=(b_1,\ldots,b_n)$ be a
sequence of positive integers such that $b_i\leq deg_G(i)$. A
subgraph $M$ of $G$ is called a \emph{$\mrm{b}$-matching} (\emph{perfect $\mrm{b}$-matching})
 {if the} degree of each vertex $i$ in $M$ is at most $b_i$ (equal to $b_i$).
Denote the set of $\mrm{b}$-matchings (perfect $\mrm{b}$-matchings)
of $G$ by $\M_G(\mrm{b})$ ($\PM_G(\mrm{b})$), and assume that it is non-empty.
 {Clearly} $\PM_G(\mrm{b})\subset \M_G(\mrm{b})$.

The weight of {a (perfect or non-perfect)} $\mrm{b}$-matching
$M$, denoted by $W_M$, is defined by
$W_{M}=\sum_{ij}w_{ij}1_{\{i,j\}\in M}$. In the next two
sections, we will restrict ourselves to the case of \emph{\it
perfect} $\mrm{b}$-matchings. The analysis is extended to
(possibly non-perfect) $\mrm{b}$-matchings in \cite{BBCZ07-arXiv}. The minimum
weight perfect $\mrm{b}$-Matching ($\mrm{b}$-MWPM), $M^*$, is
defined by $M^*=\textrm{argmin}_{M\in \PM_G(\mrm{b})}\ W_{M}$.
The goal of this paper is to find $M^*$ via a min-sum belief
propagation algorithm. Throughout the paper, we will assume
that $M^*$ is unique.

{\bf Linear Programming Relaxation.}
Assigning variables $x_{ij}\in\{0,1\}$ to the edges in
$E$, we can express the   weighted perfect $\mrm{b}$-matching problem
as the problem of finding a vector
$\mrm{x}\in\{0,1\}^{|E|}$ that minimizes the total weight
$\sum_{ij\in E}x_{ij}w_{ij}$, subject to the constraints
$\sum_{{j\in N(i)}}x_{ij}=b_i$ for all $i\in V$.  Relaxing the constraint
that $x_{ij}$ is integer, this leads to the following linear program and its dual:
\begin{equation}
\begin{array}{rcclcrccl}
&&&&&&&&\\
\textrm{min }&&\sum_{\{i,j\}\in E}x_{ij}w_{ij}&&|&\textrm{max }&&\sum_{i=1}^n b_iy_i-\sum_{{\{i,j\}\in E}} \la_{ij}&\\
\textrm{subject to}&&&&|&\textrm{subject to}&&&\\
&&\sum_{j\in N(i)} x_{ij} = b_i&\forall~~ i&|&&&w_{ij}+\la_{ij}\geq y_i + y_j&\forall~~ \{i,j\}\in E\\
&&0\leq x_{ij} \leq 1&\forall~\{i,j\}\in E&|&&&\la_{ij}\geq0&\forall~\{i,j\}\in E\\
&&&&|&&&&\\
&&&&|&&&&\\
&&\textrm{Primal LP}&&|&&&\textrm{Dual LP}&\\
\end{array}
\label{eq:LP-Relax}
\end{equation}
We say the LP relaxation \eqref{eq:LP-Relax} has \emph{no fractional solution} if,
every optimal solution $x$ of LP satisfies $x\in\{0,1\}^{|E|}$.
Note that absence of fractional solutions implies uniqueness of
integer solutions,
since any convex combination of two integer solutions is a solution to the
LP as well. We want to show that the
BP algorithm for our problem converges to the
correct solution, provided the LP relaxation (\ref{eq:LP-Relax}) has no fractional solution.

{\bf Complementary Slackness Conditions.} Complementary slackness
for  the LP and its dual state that the variables $\mrm{x}^{{*}}=(x_{ij}^*)$ and
$\mrm{y}^*=(y_i^*),~\mrm{\la}^{{*}}=(\la_{ij}^*)$ are optimum solutions to the
LP relaxation and its dual (\ref{eq:LP-Relax}),
respectively, if and only if for all edges $\{i,j\}$ of $G$ both $x_{ij}^*(w_{ij}+\la_{ij}^*-y_i^*-y_j^*)=0$.
and $(x_{ij}^*-1)\la_{ij}^*=0$ hold.
Using the fact that the {LP has no fractional solution}, one can deduce the following modified
complementary slackness conditions: For all $\{i,j\}\in M^*;$ $w_{ij}+\la_{ij}^*=y_i^*+y_j^*$ and for all $\{i,j\}\notin M^*;$ $\la_{ij}^*=0$.

By these conditions and the fact that  $\la_{ij}^*\geq 0$, we have {that}
$w_{ij}\leq y_i^*+y_j^*$ for all $\{i,j\}\in M^*$, and $w_{ij}\geq y_i^*+y_j^*$ for all $\{i,j\}\notin M^*$. However, as the counterexample given in \cite{BBCZ07-arXiv} shows, it is in general not true that these inequalities are strict even when the LP has no fractional solution.
Let $S$ be the set of those edges in $G$ for which $|w_{ij}-y_i^*-y_j^*|>0$. We will assume the minimum gap is
$\epsilon$. i.e. $\epsilon = \min_{\{i,j\}\in S}~|w_{ij} - y_i^*-y_j^*|>0$. Throughout this paper we assume that there exist an edge in $G$ for which the strict inequality $|w_{ij}-y_i^*-y_j^*|>0$ holds and therefore $\epsilon>0$ is well defined. The other cases,
where for each $\{i,j\}\in E$ the equality $w_{ij}=y_i^*+y_j^*$ holds,
happens only for special cases are discussed in \cite{BBCZ07-arXiv}  {and can be treated similarly.}
Let also \begin{equation}
\label{eq:y-bd} L = \max_{1\leq i\leq n}~|y_i^*| .
\end{equation}

{\bf Algorithm and Main Result.}
The following algorithm is  a synchronous implementation of BP for
finding the minimum weight perfect $\mrm{b}$-matching ($\mrm{b}$-MWPM). The main intuition behind this
algorithm (and, indeed, all BP algorithms) is that each vertex of the graph
assumes the graph has no cycles, and makes the best (greedy) decision
based on  {this} assumption.
Before applying the BP algorithm, we remove all \emph{trivial} vertices from the graph.
A vertex $i$ is called trivial if $deg_G(i)=b_i$. This is because all of the edges adjacent to $i$ should be in every perfect $\mrm{b}$-matching. Therefore the graph can be simplified by removal of all trivial vertices and their adjacent edges.
\vspace{.1in}
\hrule
\noindent{\bf Algorithm Sync-BP.}
\vspace{2mm}
\hrule
\begin{itemize}
\item[(1)] At times $t=0,1,\ldots$, each vertex sends  {real-valued} messages
to each of its neighbors.  {The} message of $i$ to
$j$ at time $t$ is denoted by $m_{i\to j}(t)$.

\item[(2)] Messages are initialized by $m_{i\to j}(0)
= w_{ij}$ for all $\{i,j\}\in E$ (the messages can actually be initialized by any arbitrary values \cite{{BBCZ07-arXiv}})

\item[(3)] For $t \geq 1$, messages in iteration $t$ are obtained from
messages  {in} iteration $t-1$ recursively as follows:
\begin{eqnarray}
\forall~\{i,j\}\in E:~~~~m_{i\to j}(t) & = & w_{ij} -
b_i^{th}\textrm{-min}
_{\ell\in N(i)\backslash\{j\}}\bigg[
m_{\ell\to i}(t-1)\bigg] \label{l:recursesim}
\end{eqnarray}
where $k^{th}$-min$[A]$  {denotes} the $k^{th}$ minimum\footnote{Note that the $b_i^{th}\textrm{-min}_{\ell\in N(i)\backslash\{j\}}$ is well defined since we assumed that all trivial vertices are removed and thus there are at least $b_i+1$ elements in the set $N(i)$ for each $i$.} of set A.

\item[(4)] The estimated $\mrm{b}$-MWPM at the end of iteration $t$ is $M(t)=\cup_{i=1}^nE_i(t)$ where $E_i(t)=\big\{\{i,j_1\},\ldots,\{i,j_{b_i}\}\big\}$  is such
that $N(i)=\{j_1,j_2,\ldots,j_{deg_G(i)}\}$ and $m_{{j_1}\to i}(t)
\leq m_{{j_2}\to i}(t)\cdots\leq m_{{j_{deg_G(i)}}\to i}(t)$.
i.e.,  among all $i$'s neighbors, choose edges to the $b_i$
neighbors that transfer
the smallest incoming messages to $i$.

\item[(5)] Repeat (3)-(4)  {until} $M(t)$
    converges\footnote{{ The subgraph $M(t)$ is not
    necessarily a perfect $\mrm{b}$-matching of $G$ but we
    will show that after $O(n)$ iterations it will be the
    minimum weight perfect $\mrm{b}$-matching.}}.
\end{itemize}
\vspace{2mm}
\hrule
In Lemma \ref{cor:bp-solves-tree}, we will show the main intuition behind
the equation \eqref{l:recursesim} and how it is derived. But we note that one can also
use the graphical model representations of \cite{BSS05}, \cite{HuJ07},
\cite{San07} to obtain the standard BP equations for this problem,
which, after some algebraic calculations, yield the recursive
equation \eqref{l:recursesim}.

{The main result of the paper  says that the above algorithm, which is designed for graphs with no
cycle (i.e., for trees), works correctly for a much larger family of graphs
including those with many short cycles.}

\begin{theorem}\label{thm:main}
Assume that {the LP relaxation (\ref{eq:LP-Relax})
has no fractional solution}.  {Then} the algorithm Sync-BP converges to $M^*$ after
at most $\lceil \frac{2nL}{\epsilon}\rceil$ iterations.
\end{theorem}
{If the LP relaxation \eqref{eq:LP-Relax} has a fractional solution
whose cost is strictly less than $W_{M^*}$, then \cite{San07}, \cite{SMW07} have shown
for the case of 1-matching that BP does not converge to $M^*$.
It is straightforward to generalize this to
perfect $\mrm{b}$-matching as well.
But for the case in which the LP relaxation has
a fractional solution whose cost is equal to $W_{M^*}$, BP fails
in general. This is because the $b_i^{th}$ minimum in
equation \eqref{l:recursesim} is not unique, and one needs
an oracle to make the right decision. If such an oracle {exists},
then BP converges to $M^*$.}

In what follows we first  display  the connection between the Sync-BP equations and the so-called computation
tree. Next we discuss the how  the complementary slackness conditions is related to alternating paths in the graph $G$. These results are eventually used  to prove that, when the LP relaxation has no fractional solutions, then
solutions on the computation tree are the same as the solutions on the original graph $G$.

 {
{\bf Analysis of the Synchronous BP via Computation Tree.}} The main idea behind the algorithm Sync-BP is that it assumes the graph $G$ has no
cycle. In other words, it finds the $\mrm{b}$-MWPM of a graph $G'$
that has the same local structure as $G$ but no cycles. The precise definition of the computational tree  for Sync-BP goes as follows.

{\bf Computation Tree.} For any $i\in V$, let $T_{i}^t$ be the $t$-level computation
tree corresponding to $i$, defined as follows: $T_{i}^{t}$ is a weighted tree of height $t+1$,
{rooted at $i$}. All tree-nodes have labels from the set $\{1,\ldots,n\}$ according to the following
recursive rules: (a) {The root} has label $i$, (b) The {set of labels of the} $deg_G(i) $ children of the root is
equal to $N(i)$, and (c) If $s$ is a non-leaf node whose parent has label $r$,
then the set of labels of its children is $N(s)\backslash\{r\}$.
$T_{i}^{t}$ is often called the {\em
{unwrapped} tree} at node $i$. The computation tree is well known technique for analyzing algorithms and
constructed by replicating the local connectivity of the original
graph. The messages received by node
$ i$ in the belief propagation algorithm after $t$ iterations in graph
$G$ are equivalent to those that would have been received by the root $ i$
in the computation tree, if the messages were passed up along the
tree from the leaves to the root.
A subtree $\mc{M}$ of edges in the computation tree $T_{
i}^{t}$ is called a perfect \emph{tree-$\mrm{b}$-matching} if
for each \emph{non-leaf} vertex with label $i$ we have
$deg_{\mc{M}}( i)=b_i$. Now denote the minimum weight perfect
tree-$\mrm{b}$-matching ($\mrm{b}$-TMWPM) of the computation
tree $T_{i}^{t}$ by $\mc{N}^*(T_{ i}^{t})$. The following lemma shows that
Sync-BP can be seen as a dynamic programming procedure that
finds the minimum weight perfect tree-$\mrm{b}$-matching over
the computation tree. Figure \ref{fig:one} shows a graph $G$
and one of its corresponding computation tree.

\begin{figure}
\centering
      \includegraphics[scale=0.4]{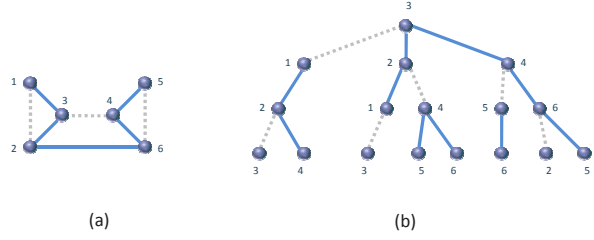}
      \caption{ Part (a) shows a graph $G$ where dashed and gray edges represent a $1$-matching. Part (b) shows the computation tree $T_3^2$ corresponding to $G$ where the set of dashed and gray edges form a $1$-TMWPM.}
      \label{fig:one}
\end{figure}
\begin{lemma}\label{cor:bp-solves-tree}
The algorithm Sync-BP solves the $\mrm{b}$-TMWPM problem on the computation
tree. In particular, for each vertex $ i$ of $G$, the set of
$E_i(t)$ {which} was chosen at the end of iteration $t$ by
Sync-BP is exactly the set of $b_i$ edges {which} are attached to the
root in $\mrm{b}$-TMWPM of $T_{ i}^{t}$.
\end{lemma}
Lemma \ref{cor:bp-solves-tree} characterizes the estimated $\mrm{b}$-MWPM,
 $M(t)$, and will be used later in the proof of the main result.

The following technical Lemma is also crucial for the proof of equivalence between BP and LP.   It  connects the complementary slackness conditions   to paths on the graph $G$ and on the computation tree. This lemma  provides the connection between the absence of fractional
solutions and the correctness of BP.

\begin{definition}\label{def:alt-path}
A path $P=( {i_1}, {i_2},\ldots, {i_k})$ in $G$ is called
\emph{alternating path} if it has the following two properties: \textbf{(i)} There exist a partition {of} edges of
$P$ into two sets $A,B$ such that either $(A\subset M^*~,~B\cap M^*=\emptyset)$ or $(A\cap M^*=\emptyset~,~B\subset M^*)$.
Moreover $A$ ($B$) consists of all odd (even) edges; i.e. $A=\{({i_1}, {i_2}), ( {i_3}, {i_4}),\ldots\}$
($B=\{( {i_2}, {i_3}), ( {i_4}, {i_5}),\ldots\}$), and \textbf{(ii)} The path $P$ might intersect itself or even
repeat its own edges but no edge is repeated immediately. That is, for any $1\leq r\leq k-2:~ i_r\neq  i_{r+1}$ and $ i_r\neq  i_{r+2}$.
$P$ is called an \emph{alternating cycle} if $ {i_1}= {i_k}$.
\end{definition}
\begin{lemma}\label{lem:path-2n}
Assume that the LP relaxation (\ref{eq:LP-Relax}) has no fractional solution. Then for any alternating path $P$ of length at least $2n$, there
exists an edge $\{i,j\}\in P$ such that
{the} inequality $|w_{ij} -
y_i^*-y_j^*|>0$ holds. That is, $P\cap S\neq\emptyset$.
\end{lemma}
{\bf Proof of Theorem \ref{thm:main}}
We will prove Theorem \ref{thm:main}, namely that {if}
 {the LP relaxation (\ref{eq:LP-Relax})  has no fractional solution and hence $M^*$ is unique}, then  Sync-BP
converges to the correct $\mrm{b}$-MWPM. We will do this by
showing that if the depth of computation tree is large enough, then for
any vertex $i$, its neighbors in $M^*$ ($\mrm{b}$-MWPM of $G$) are
exactly those children that are selected in $\mc{N}^*(T_{ i}^{t})$
($\mrm{b}$-TMWPM of $T_{ i}^{t}$). Here is the main lemma that
summarizes the above claim:
\begin{lemma}\label{lem:tree=graph}
If the LP relaxation (\ref{eq:LP-Relax}) has no fractional solution, then
for any
vertex~$i$ of $G$ and for any $t> \frac{2nL}{\epsilon}$, the set of
edges that are adjacent to root $i$ in $\mc{N}^*(T_{ i}^t)$ are
exactly those edges that are connected to $ i$ in $M^*$.
\end{lemma}
The proof of Lemma \ref{lem:tree=graph} is the main technical part of
this work. The high level overview of the underlying argument goes as follows.  Consider the computation tree
($T_{ i}^t$) rooted at vertex $ i$ and look at
$\mc{N}^*(T_{ i}^t)$. We will assume that the claim of the lemma
does not hold. That is, we assume that at the root,
$\mc{N}^*(T_{ i}^t)$ does not choose the same edges as $M^*$-edges
adjacent to $ i$. Then we use the property of perfect tree-$\mrm{b}$-matchings,
{namely}
that each non-leaf vertex $ j$ is connected to exactly $b_j$ of its
neighbors, to construct a new perfect tree-$\mrm{b}$-matching on the computation
tree. This new perfect tree-$\mrm{b}$-matching is going to have less total
weight if the depth of the computation tree is large enough. This last
step uses an alternating path argument which is a highly non-trivial generalization
of the technique of \cite{BSS05} for the case of perfect $1$-matching in
bipartite graphs. For this part we will use the solutions to the dual LP
(\ref{eq:LP-Relax}).

\begin{proof}[Proof of Lemma \ref{lem:tree=graph}]
Let us denote the lifting of a perfect $\mrm{b}$-matching $M^*$ to a perfect tree-$\mrm{b}$-matching
on $T_{ i}^t$ by $\mc{M}^*$. That is, $\mc{M}^*$
consists of all edge of the computation tree with endpoint labels
$ i, j$ such that $\{i,j\}\in M^*$ as an edge in $G$. The goal
is to show that $\mc{N}^*(T_{ i}^t)$ and $\mc{M}^*$ have the same
set of edges at the root of the computation tree. To lighten the notation,
we denote the  $\mrm{b}$-TMWPM of  $T_{ i}^t$  {by} $\mc{N}^*$.

Assume the contrary, that there exist children $ {i_{-1}}, {i_1}$
of root $ i$ such that $\{ i, {i_1}\}\in
\mc{M}^*\backslash\mc{N}^*$ and $\{ i, {i_{-1}}\}\in
\mc{N}^*\backslash\mc{M}^*$. Since both $\mc{M}^*,~\mc{N}^*$ are
perfect tree-$\mrm{b}$-matchings, they have $b_{i_1}$ edges connected to
$i_1$. Therefore there exist a child $ {i_2}$ of $ {i_1}$ such
that $\{ {i_1}, {i_2}\}\in \mc{N}^*\backslash\mc{M}^*$. Similarly
there is a child $ {i_{-2}}$ of $ {i_{-1}}$ such that
$\{ {i_{-1}}, {i_{-2}}\}\in \mc{M}^*\backslash\mc{N}^*$. Therefore
we can construct a set of alternating paths $P_{\ell}, ~\ell \geq 0$,
in the computation tree, that contain edges from $\mc{M}^*$ and
$\mc{N}^*$ alternatively defined as follows. Let $ {i_0} =
\mbox{root}~ i$ and $P_0 = ( {i_0})$ be a single vertex
path. Let $P_1 = ( {i_{-1}},  {i_0},  {i_1})$, $P_2 =
( {i_{-2}},  {i_{-1}},  {i_0},  {i_1},  {i_2})$ and similarly
for $r \geq 1$, define $P_{2r+1}$ and $P_{2r+2}$ recursively as
follows: $P_{2r+1} = ( {i_{-(2r+1)}}, P_{2r},  {i_{2r+1}})$, $
P_{2r+2} = ( {i_{-(2r+2)}}, P_{2r+1},  {i_{2r+2}})$
where $ {i_{-(2r+1)}},~ {i_{2r+1}} $ are nodes at level $2r+1$
such that $\{ {i_{2r}}, {i_{2r+1}}\}\in
\mc{M}^*\backslash\mc{N}^*$ and $\{ {i_{-2r}}, {i_{-(2r+1)}}\}
\in \mc{N}^*\backslash\mc{M}^*$. Similarly
$ {i_{-(2r+2)}},~ {i_{2r+2}} $ are nodes at level $2r+2$ such
that
$\{ {i_{2r+1}}, {i_{2r+2}}\}\in
\mc{N}^*\backslash\mc{M}^*$ and $\{ {i_{-(2r+1)}}, {i_{-(2r+2)}}\}
\in \mc{M}^*\backslash\mc{N}^*$.
Note that, by definition, such paths $P_{\ell}$ for $0\leq \ell \leq
t$ exist since the tree $T_{ i}^t$ has $t+1$ levels and can support
a path of length at most $2t$ as defined above.
Now consider the path $P_t$ of length $2t$. It is an alternating
path on the computation tree with edges from $\mc{M^*}$ and $\mc{N}^*$.
Let us refer to the edges of $\mc{M^*}$ ($\mc{N}^*$) as the
$\mc{M^*}$-edges ($\mc{N}^*$-edges) of $P_t$.
We will now modify the perfect tree-$\mrm{b}$-matching $\mc{N}^*$ by
replacing all $\mc{N}^*$-edges of $P_t$ with their complement in
$P_t$ ($\mc{M^*}$-edges of $P_t$). It is straightforward that this
process produces a new perfect tree-$\mrm{b}$-matching $\mc{N}'$ in
$T_{ i}^t$.

Let us assume, for the moment, the following lemma:
\begin{lemma}\label{lem:switch}
The weight of the perfect tree-$\mrm{b}$-matching $\mc{N}'$ is strictly less
than that of $\mc{N}^*$ on $T_{ i}^t$.
\end{lemma}
This completes the proof of Lemma \ref{lem:tree=graph} since Lemma
\ref{lem:switch} shows that $\mc{N}^*$ is not the minimum weight
perfect tree-$\mrm{b}$-matching on $T_{ i}^t$, leading to a
contradiction.\enp
\end{proof}
\begin{proof}[Proof of Lemma \ref{lem:switch}] It suffices to show that the total
weight of the $\mc{N}^*$-edges of $P_t$ is more than the total
weight of $\mc{M^*}$-edges of $P_t$. For each vertex $ {i_r}\in
P_t$ consider the value $y_{i_r}^*$ from the optimum solution to the dual LP
(\ref{eq:LP-Relax}). Using the inequality $w_{ij}\leq y_i^*+y_j^*$ for
edges of $\mc{M}^*$, we obtain: $\sum_{\{i,j\}\in P_t\cap \mc{M}^*}w_{ij} \leq
\left(\sum_{r=-t}^{t}y_{i_r}^*\right) - y_{i_{(-1)^tt}}^*-k_1\ep$ where $k_1$ is the number of $\mc{M}^*$-edges of $P_t$ that belong to
$S$, {i.e.,} the number of $\mc{M}^*$-edges of
$P_t$ endowed with the strict
inequality $w_{ij}\leq y_i^*+y_j^*$, with a gap of at least $\ep$. On
the other hand, using the inequality $w_{ij}\geq y_i^*+y_j^*$ for edges
of $\mc{N}^*$ we have: $\sum_{\{i,j\}\in P_t\cap \mc{N}^*}w_{ij}\geq
\left(\sum_{r=-t}^{t}y_{i_r}^*\right) - y_{i_{(-1)^{t+1}t}}^*+k_2\epsilon$ where now $k_2$ is number of $\mc{N}^*$-edges of $P_t$ that belong to
$S$, or equivalently the number of times the inequality $w_{ij}\geq
y_i^*+y_j^*$ is strict with a gap of at least $\ep$. One finds
$$
\sum_{\{i,j\}\in P_t\cap \mc{N}^*}w_{ij}
-\sum_{\{i,j\}\in P_t\cap \mc{M}^*}w_{ij}
 = y_{i_{(-1)^tt}}^*- y_{i_{(-1)^{t+1}t}}^*+(k_1+k_2)\epsilon
\nonumber\\
\stackrel{(a)}{\geq} (k_1+k_2)\epsilon-2L\stackrel{(b)}{\geq} (k_1+k_2)\epsilon-2L\stackrel{(c)}{>}0.
$$

Here $(a)$ uses definition of $L$
from eq. \ref{eq:y-bd}
 and
$(b)$ uses the fact that for all $i,j:~\la_{ij}^*\geq 0$. The main step is $(c)$, which
uses Lemma \ref{lem:path-2n} as follows. Path $P_t$ has length $2t$,
and each continuous piece of it with length $2n$ has a projection to
the graph $G$ which satisfies the  conditions of Lemma \ref{lem:path-2n}.
This means the path has at least one edge from the set $S$. Thus
$(k_1+k_2)\geq \frac{2t}{2n}>\frac{2L}{\epsilon}$. This completes
the proof of Lemma \ref{lem:switch}.\enp
\end{proof}

{\bf Extension to Possibly Non-Perfect $\mrm{b}$-Matchings.}
Here we note that the algorithm and the results
of the previous sections can be easily generalized to
the case of $\mrm{b}$-matchings (subgraphs $H$ of $G$
such that degree of each vertex $i$ in $H$ is {\emph{at most}} $b_i$).
Let $U(H)\subset V$ be the set of \emph{unsaturated} vertices of $G$
(vertices $i\in V$ such that $deg_H(i)<b_i$).
The minimum weight
$\mrm{b}$-Matching ($\mrm{b}$-MWM), $H^*$,
is the $\mrm{b}$-Matching such that $H^*=\textrm{argmax}_{H\in \M_G(\mrm{b})}\ W_{H}$. Note that $H^*$ does not include any edge with positive weight
because removing such edges from $H^*$ reduces its weight while
keeping it a $\mrm{b}$-matching.

{\bf Asynchronous BP.}
In the remaining we study the asynchronous version of the BP algorithm.
The update equations are exactly analogous to the synchronous
version, but at each time only a subset of the edges are updated
in an arbitrary order.
Consider the set $\vec{E}$ of all directed edges in the $G$; i.e.,
 $\vec{E}=\{(i\to j)~~s.t.~~~ i\neq j\in V\}$.
Let $A$ be a sequence $\vec{E}(1),\vec{E}(2),\ldots$ of subsets of the set $\vec{E}$. Then the asynchronous BP algorithm corresponding to the sequence $A$ can be obtained by modifying only the update rule in the step (3) of the algorithm Sync-BP as follows:
$m_{i\to j}(t) =  w_{ij} - b_i^{th}\textrm{-min}_{ \ell\in N( i)\backslash\{ j\}}\bigg[ m_{ \ell\to  i}(t-1)\bigg]$~~
if $(i\to j)\in\vec{E}(t)$ and if $(i\to j)\notin\vec{E}(t)$ then the message will not be updated, i.e. it remains equal to $m_{i\to j}(t-1)$.

This is the most general form of the asynchronous BP and it includes the synchronous version ($\vec{E}(t)=\vec{E}$
for all $t=1,2,\ldots$) {as a special case}. In many applications, a special case
of the asynchronous BP is used for which each set
$\vec{E}(t)$ consists of a single {element.}

We assume that the sequence $A$ of the updates does not have \emph{redundancies}.
That is, no edge direction $(i\to j)\in \vec{E}$ is re-updated before at least one of its {incoming} edge directions ($(\ell\to i)$ for
$\ell\in N(i)\backslash\{j\}$) is updated. More formally, if $(i\to j)\in \vec{E}(t)\cap\vec{E}(t+s)$ and
$(i\to j)\notin \cup_{r=1}^{s-1}\vec{E}(t+r)$, then at least for one $\ell\in N(i)\backslash\{j\}$,
we should have $(\ell\to i)\in \cup_{r=1}^{s-1}\vec{E}(t+r)$.

Let us denote the above algorithm by Async-BP. We claim that, if each edge
direction $(i\to j)\in\vec{E}$ is updated $\theta(n)$ times, then the
same result as Theorem \ref{thm:main} can be {proved} here.
That is, let $u(t)$ be the minimum number of times that an edge direction of the {graph} $G$ appears in the
sequence $\vec{E}(1),\ldots,\vec{E}(t)$; i.e., $u(t) = \min_{(i\to j)\in\vec{E}}\Big(\bigg|\big\{\ell:~~s.t.~~ 1\leq\ell\leq
t~~~\textrm{and}~~~(i\to j)\in\vec{E}(\ell)\big\}\bigg|\Big)$.
From the definition, $u(t)$ is a non-decreasing function of $t$. We claim that the following result holds:
\begin{theorem}\label{thm:async-b-match}
Assume that the LP relaxation (\ref{eq:LP-Relax})
has no fractional solution.
Then the algorithm Async-BP
converges to $M^*$ after at most $t$ iterations, provided
$u(t)> \frac{2nL}{\epsilon}$.
\end{theorem}
Proof of the above theorem relies on the notion of generalized computation tree for the asynchronous version of the BP algorithm which will be  given in the longer version of this paper \cite{BBCZ07-arXiv}.

Finally we note that the same algorithm as Async-BP and the same result
as Theorem \ref{thm:async-b-match} can be stated and proved for the (possibly non-perfect)
$\mrm{b}$-matchings as well.

\section{Acknowledgements}\label{sec:ack}
We would like to thank L\'{a}szl\'{o} Lov\'{a}sz, Andrea Montanari,
Elchannan Mossel and Amin Saberi for useful discussions.
{This work was done while Riccardo Zecchina was a Visiting Researcher in the Theory
Group at Microsoft Research, and was supported by the Microsoft Technical Computing
Initiative.}


{\small
\begin{thebibliography}{1}


\bibitem{MeZ02}
M. Mezard and R. Zecchina ``Random K-satisfiability: from an analytic solution to a new efficient algorithm,'' \emph{Phys.Rev. E E},  66, 056126, 2002.

\bibitem{BMZ05}
A.~Braunstein, M.~Mezard, and R.~Zecchina, ``Survey propagation: an
algorithm   for satisfiability,'' \emph{Random Structures and Algorithms}, vol.~27, pp.
  201--226, 2005.


\bibitem{SMW07}
S. Sanghavi, D. Malioutov, A. Willsky
``Linear programming analysis of loopy belief propagation for weighted matching'',
to appear in \emph{NIPS}, 2007.

\bibitem{BSS05}
M.~Bayati, D.~Shah, and M.~Sharma, ``Maximum weight matching via
max-product
  belief propagation,'' \emph{Preliminary version appeared at IEEE ISIT 2005.
  Longer version {\em to appear in} IEEE\ Trans.\ Information\ Theory}, 2007.

\bibitem{HuJ07}
B. Huang, T. Jebara, ``Loopy belief propagation for bipartite
maximum weight b-matching'', \emph{Artificial Intelligence and
Statistics (AISTATS)}, March, 2007.

\bibitem{Ger95}
A. Gerards, ``Matching. Volume 7 of ,'' \emph{Hand book of Operation Research and Management Science}, Chapter 3, pp. 135-224. North-Holland, 1995.

\bibitem{Pul95}
W. Pulleyblank, ``Matchings and extensions.'' Volume 1 of \emph{Handbook of Combinatorics}, Chapter 3, pp. 179-232, North Holland, 1995.

\bibitem{MeP86}
M. Mezard and G. Parisi ``Mean-field equations for the matching and travelling Salesman problems,'' \emph{Eurhophysics letters},
Vol. 2, pp. 913-918, 1986.

\bibitem{Ald01}
D. Aldous, ``The zeta (2) Limit in the Random Assignment Problem,''
\emph{Random Structures and Algorithms}, Vol.\ 18, pp.~381-418, 2001.


\bibitem{FWK05}
J. Feldman, M. Wainwright and D. Karger, ``Using linear programming
to decode binary linear codes'', \emph{IEEE Transactions on
Information Theory}, vol. 51, pp. 954-972, 2005.

\bibitem{VoK04} P.O. Vontobel and R. Koetter, ``On the relationship between linear programming decoding and min-sum algorithm decoding,'' \emph{Proc. ISITA} 2004, Parma, Italy, pp. 991--996, Oct. 10-13, 2004.

\bibitem{VoK05} P.O. Vontobel and R. Koetter, ``Graph-cover decoding and finite-length analysis of message-passing iterative decoding of LDPC codes,'' to apprear in \emph{IEEE Trans. Inform. Theory}, http://www.arxiv.org/abs/cs.IT/0512078.

\bibitem{WJW05}
M. Wainwright, T. Jaakkola, and A. Willsky, ``MAP estimation via agreement on trees: message-passing
and linear programming'', {\em IEEE Transactions on Information Theory},
51(11):3697:3711, 2005.

\bibitem{WYM07}
Y. Weiss, C. Yanover and T. Meltzer ``MAP Estimation, Linear
Programming and Belief Propagation with Convex Free Energies,''
\emph{UAI},  2007.

\bibitem{MoV07J}
C. Moallemi and B. Van Roy, ``A Message-Passing Paradigm for Resource Allocation,'' preprint June 2007.

\bibitem{BSS06}
M.~Bayati, D.~Shah, and M.~Sharma, ``Max-product for maximum weight matching: convergence, correctness and LP duality,'' in \emph{IEEE Int.\ Symp.\ Information\ Theory}, 2006.

\bibitem{Ber88}
D.~P. Bertsekas, ``The auction algorithm: A distributed relaxation
method for  the assignment problem,'' \emph{Annals of Operations Research}, vol.~14, 1988.

\bibitem{GNS05}
D. Gamarnik,  T. Nowicki and G. Swirscsz, ``Maximum Weight Independent Sets and Matchings in Sparse Random Graphs. Exact Results using the Local Weak Convergence Method'',  \emph{Random Structures and Algorithm}, Vol.28, No. 1, pp. 76-106, 2005.


\bibitem{ZdM06}
L. Zdeborov\'{a} and M. M\'{e}zard, ``The number of matchings in random graphs'', \emph{J. Stat. Mech.}, 2006.

\bibitem{BBCZ07-arXiv}
M. Bayati, C. Borgs, J. Chayes and R. Zecchina, ``Belief-Propagation for Weighted $\mrm{b}$-Matchings on Arbitrary Graphs and its
Relation to Linear Programs with Integer Solutions'', \emph{in arXiv},  http://www.arxiv.org/abs/0709.1190v2, September 8, 2007.

\bibitem{San07} S. Sanghavi, ``Equivalence of LP Relaxation and Max-Product for Weighted Matching in General Graphs'', \emph{IEEE Information Theory Workshop}, September 2007, perliminary version arXiv:0705.0760, May 5 2007.






















%
%
%
%
%
%
%
%
%
%
%
%
%





\end{thebibliography}
}
\end{document}